\documentclass[11pt]{article}
\usepackage[left=1in,top=1in,right=1in,bottom=1in,head=0in,nofoot]{geometry}
\usepackage{setspace,url,bm,amsmath}
\usepackage{scalerel}
\usepackage{xcolor}

\usepackage{multirow}
\usepackage{arydshln}
\usepackage{float}
\usepackage{booktabs}
\usepackage{lscape}

\usepackage{mathtools}
\usepackage{amsmath}
\usepackage{amssymb}
\usepackage{amsthm}
\usepackage{bbm}
\usepackage{latexsym}
\theoremstyle{definition}
\newtheorem{assumption}{Assumption}
\newtheorem*{theorem*}{Theorem}
\newtheorem{theorem}{Theorem}
\newtheorem*{rmk*}{remark}
\newtheorem{proposition}{Proposition}
\newtheorem{lemma}{Lemma}

\newtheorem{corollary}{Corollary}
\newtheorem*{corollary*}{Corollary}

\usepackage{graphicx}
\usepackage{caption}
\usepackage[margin=20pt]{subcaption}
\graphicspath{ {./figs/} }

\usepackage{enumerate}
\usepackage{titlesec}
\titlelabel{\thetitle.\quad}
\titleformat*{\section}{\bf\Large\center}

\usepackage{authblk}

\usepackage{natbib} 
\bibpunct{(}{)}{;}{a}{}{,}

\usepackage{multibib}
\newcites{sec}{References}

\usepackage{listings}
\usepackage{hyperref}
\usepackage{comment}

\def \intd{\mathrm{~d}}
\newcommand{\indep}{\perp \!\!\! \perp}

\usepackage{etoolbox}
\apptocmd{\sloppy}{\hbadness 10000\relax}{}{}
\RequirePackage[normalem]{ulem}
\allowdisplaybreaks
\usepackage{makecell}
\usepackage{steinmetz}
\usepackage{colortbl}

\begin{document}
 
\singlespacing

\title{\bf Probabilities of Causation for Continuous Outcomes: Bounds and Identification} 

\author[1]{Jile Chaoge}
\author[1]{Kesen Han}
\author[1]{Fahui Liu\thanks{Corresponding author: 2451051003@st.btbu.edu.cn}}
\author[1]{Peng Wu}
\affil[1]{\small School of Mathematics and Statistics, Beijing Technology and Business University, 100048, China}

\date{}

\maketitle

\begin{abstract}
The probability of necessity (PN), which quantifies the probability that an observed event would not have occurred in the absence of the treatment, is a central estimand in attribution analysis. While PN has been extensively studied for binary outcomes and has recently been developed for ordinal outcomes, a formal framework for continuous outcomes remains underdeveloped. To address this gap, we propose the general probability of necessity (GPN) for continuous outcomes, a setting that is substantially more challenging than the binary and ordinal cases. Rather than imposing strong identifiability assumptions, we adopt a partial identification perspective and derive sharp lower and upper bounds under standard assumptions of ignorability and monotonicity. We further introduce a copula-based framework that exploits dependence information between potential outcomes to tighten these bounds. Simulation studies and real-world applications demonstrate the effectiveness of our method. 
\end{abstract}


\newpage
 
\onehalfspacing

\section{Introduction}
\label{intro}

Much of the causal inference literature has focused on forward-looking questions, such as estimating the average treatment effect 
on a population. 
In many high-stakes domains, however, the central questions are backward-looking, or attributional~\cite{Pearl2015, Dawid2022}.   
For example, in legal settings, courts must determine whether a specific plaintiff's injury was caused by the defendant's action, rather than  assessing general risks~\cite{dawid2014fitting, sanders2021differential}. 
Similarly, in personalized medicine, it is often necessary to assess whether a patient's recovery can be attributed to a treatment or whether it would have occurred in its absence~\cite{khoury2004epidemiologic}. 
The primary estimands for such questions are the  probability of necessity (PN) and its analogues~\cite{Pearl1999}.

While PN is well studied for binary outcomes and has recently been generalized to ordinal settings~\cite{Zhang2024}, a formal formulation for continuous outcomes--such as blood pressure, tumor size, or financial loss--remains underdeveloped. In practice, discretizing continuous outcomes leads to information loss, motivating the need for an attribution framework that operates on continuous outcomes.

The main challenge in identifying the PN for continuous outcomes arises from the unobservability of the joint distribution of the potential outcomes $(Y(1), Y(0))$. Since these two outcomes are never observed simultaneously for any individual, their dependence structure is not identifiable from the observed data, even under randomization~\cite{Holland1986}. In binary and ordinal settings, this difficulty is typically addressed by assuming monotonicity ($Y(1)\geq Y(0)$) for binary outcomes and  monotonic incremental treatment effect ($0\leq Y(1) - Y(0)\leq 1$) for ordinal outcomes, achieving point identification~\cite{Pearl1999, Zhang2024}.   
In the continuous setting, however, monotonicity conditions are no longer sufficient. This is because monotonicity restricts the direction of the treatment effect but does not determine the specific value of $Y(1)$ given $Y(0)$ (or $Y(0)$ given $Y(1)$), providing only limited information. Consequently, relying on these assumptions often yields wide bounds.

In this paper, we consider a new estimand, the general probability of necessity (GPN), for continuous outcomes. The formulation of the GPN provides a unified framework for defining PN across various types of outcomes. 
Rather than imposing additional restrictive assumptions to achieve point identification, we adopt a partial identification approach that derives lower and upper bounds for the GPN.

We first derive bounds under the standard ignorability assumption~\cite{rosenbaum1983central}, as well as under the combined ignorability and monotonicity assumptions. The former corresponds to the classic Fréchet-–Hoeffding bounds. These bounds provide a basis for comparison and encompass several previous results for binary outcomes.
Subsequently, given that the monotonicity is often  strong in practice and provides limited information about the GPN, we propose a novel method to improve the Fréchet-–Hoeffding bounds based on copula models~\cite{copulas2006}. By Sklar's Theorem~\cite{sklar1959fonctions}, the joint distribution of $(Y(1), Y(0))$ can be expressed in terms of its marginal distributions and a copula capturing their dependence. Accordingly, we  reformulate the GPN in a copula representation.

The copula representation provides a key insight into how to improve bounds on the GPN. The rationale is that, under the copula representation, the GPN depends only on an association parameter $\rho$ that captures the dependence between the potential outcomes under the ignorability assumption. When no information about $\rho$ is available, the resulting bounds reduce to the Fréchet--Hoeffding bounds. However, any domain knowledge that restricts $\rho$ to a subset of its admissible range serves to tighten these bounds.

Notably, the association parameter $\rho$ is not identifiable under the ignorability assumption alone, but its plausible range may be informed by domain knowledge; see~\cite{Wu-etal-2024-Harm, Bodik2025} for further discussion. Moreover, even when specifying a precise interval for $\rho$ is difficult due to limited expert knowledge, determining the sign of the dependence is often feasible. A simple and practically appealing restriction is to assume positive dependence (i.e., $\rho \geq 0$), which yields significantly tighter  bounds on GPN. 
Moreover, the copula representation naturally establishes a sensitivity analysis framework for $\textup{GPN}$ by treating $\rho$ as the sensitivity parameter, enabling the assessment of the robustness of the resulting conclusions.

The main contributions are summarized as follows: 
\begin{itemize}
	\item  We propose a new estimand for attribution analysis for continuous outcomes: the general probability of necessity (GPN), which extends and encompasses previous estimands for binary and ordinal outcomes.

	\item  We derive identification bounds for the GPN under the standard ignorability assumption, as well as under the combined ignorability and monotonicity assumptions.

	\item  We further develop a framework for improving the bounds on the GPN based on copula models. This framework not only facilitates the incorporation of expert knowledge to improve bounds but also provides a sensitivity analysis method to evaluate the robustness of the bounds as the association parameter varies.

	\item  We illustrate the proposed bounds through simulations and real-world applications, demonstrating the effectiveness of the proposed method. 
\end{itemize}

\section{Related Work}
\label{related}

\textbf{Probabilities of Causation.} 
Causal attribution studies ``causes of effects'', in contrast to the estimation of ``effects of causes'' (e.g., ATE)~\cite{Dawid2022}. 
It has been applied in medicine, law, policy analysis, and invariant learning~\cite{Pearl2015, dawid2014fitting, sanders2021differential, khoury2004epidemiologic, Yang-etal2023-NIPS, Sun-etal2025-ICML}.~\cite{Pearl1999} formalized attribution within structural causal models by defining the probability of necessity (PN) and the probability of sufficiency (PS). 
Early work focused on binary outcomes without covariates, deriving identification results and sharp bounds under standard assumptions. 
Under ignorability, sharp bounds on PN were obtained, whereas the addition of monotonicity led to point identification~\cite{Pearl1999, Tian2000}. Subsequent studies incorporated covariates and mediators to further tighten these bounds~\cite{Dawid2017, Mueller2021, zhang2025confidence}.

In recent years, extensions to settings with non-binary treatment or non-binary outcomes have been considered~\cite{Li2022learning, Zhao2023conditional, Li2024probabilities}.
Regarding non-binary outcomes,~\cite{Zhang2024} extended PN and PS to ordinal outcomes and achieved identification under a monotonic incremental treatment effect assumption. 
In contrast, a general framework for continuous outcomes remains limited, despite their prevalence in scientific and economic applications. 
This work addresses this gap by defining and discussing the GPN for continuous outcomes.

\textbf{Partial Identification and Sensitivity Analysis.} 
A central difficulty in attribution analysis is that the joint distribution of potential outcomes is not identified from observational data~\cite{Holland1986}. 
As a result, point identification typically requires strong assumptions, such as monotonicity~\cite{Pearl1999} or conditional independence~\cite{Shen2013treatment}. 
Beyond identification,~\cite{Tian2025semiparametric} studied statistical efficiency and proposed semiparametric efficient estimators for PN and PS under these assumptions.
However, when such assumptions are not credible, partial identification provides an alternative by characterizing sets of compatible values. Recent work has adopted sensitivity analysis and bounding approaches that relax restrictive assumptions. 
For example,~\cite{Bodik2025} used copulas to derive prediction intervals for individual treatment effects. 
In this setting, we introduce a copula-based framework for GPN that allows the dependence structure between potential outcomes to be restricted using domain knowledge (e.g., positive dependence). 
This approach tightens the Fréchet--Hoeffding bounds without imposing monotonicity.

\section{Preliminaries}
\label{setup}

\subsection{Notation}
Let \( Z \in \{0,1\} \) denote a binary treatment indicator, \( X \in \mathcal{X} \subseteq \mathbb{R}^d \) a vector of pre-treatment covariates, and \( Y \in \mathbb{R} \) the observed outcome. We adopt the potential outcomes framework~\cite{Neyman1990, Rubin1974} to describe causal estimands. Specifically, for each unit, let \( Y(1) \) and \( Y(0) \) denote the potential outcomes that would be observed under treatment (\( Z=1 \)) and control (\( Z=0 \)), respectively. 
Under the stable unit treatment value assumption (SUTVA), the observed outcome  
$Y = Z\,Y(1) + (1-Z)\,Y(0).$
Suppose that the observed data \( \mathcal{D} = \{(X_i, Z_i, Y_i)\}_{i=1}^n \) consist of \( n \) independent and identically distributed (i.i.d.) observations drawn from a superpopulation \( \mathbb{P} \). Let \( \mathbb{E}[\cdot] \) denote the expectation with respect to \( \mathbb{P} \).

\subsection{Probabilities of Causation}
In causal inference, the goal is not only to quantify the effects of specific interventions but also to uncover the causes underlying observed outcomes~\cite{Dawid2022}; the latter is typically addressed through attribution analysis~\cite{pearl2009causality, Pearl-etal2016-primer}.
In the context of attribution analysis, much of the literature focuses on binary outcomes, for which two widely used causal quantities are the probability of necessity (PN) and the probability of sufficiency (PS), defined as
\begin{equation}  \label{eq1}
	\begin{split}
		\text{PN} ={}&  \mathbb{P}(Y(0)=0 \mid Z=1, Y=1), \\
		\text{PS} ={}& \mathbb{P}(Y(1)=1 \mid Z=0, Y=0). 
	\end{split}
\end{equation} 
Intuitively, PN evaluates whether the treatment was a necessary factor for the observed event ($Y=1$); it represents the probability that the event would not have occurred had the treatment been absent, given that the event occurred under treatment. In contrast, PS captures the sufficiency of the treatment, quantifying its capacity to produce the event in situations where it would not have occurred under control.

Recently,~\cite{Zhang2024} extended PN and PS from binary to ordinal outcomes. For example, PN was defined as $\mathbb{P}(Y(0) < y \mid Z=1, Y=y)$.
In this paper, we extend prior work to accommodate continuous outcomes. In many application domains, such as clinical trials (e.g., tumor size and blood pressure) and economics (e.g., revenue and customer lifetime value), outcomes are inherently continuous. Discretizing these variables inevitably incurs information loss and introduces sensitivity to binning choices.

\section{Conceptual Framework}
\label{concept}

In this section, we present two metrics for attribution analysis with continuous outcomes, relate them to existing formulations for binary and ordinal outcomes, and highlight the  challenges posed by the continuous setting.

For continuous outcomes, we define the general probability of necessity (GPN) and the general probability of sufficiency (GPS) as follows:  
\begin{equation}   \label{eq2}
	\begin{aligned}
		\textup{GPN} &= \mathbb{P}(Y(0) < c_0 \mid Z=1, Y \ge c_1),\\
		\textup{GPS} &= \mathbb{P}(Y(1) \ge c_1 \mid Z=0, Y < c_0),
	\end{aligned}
\end{equation}
where $c_0$ and $c_1$ are prespecified thresholds for potential outcomes $Y(0)$ and $Y(1)$, respectively. We set $c_1 \geq c_0$ for ease of interpretation, although this is not mathematically required. 
The values of $c_0$ and $c_1$ are used to define the event in which the treatment $Z$ exerts the minimal effect required to constitute an effective cause.  
To avoid redundancy, we focus on GPN in the following discussion, as the analysis for GPS can be conducted in a similar manner.

A more flexible definition of the GPN is $\mathbb{P}( \underline{c_0} \leq Y(0) < \overline{c_0} \mid Z=1, \underline{c_1} \leq Y < \overline{c_1})$, where $(\underline{c_0}, \overline{c_0}, \underline{c_1}, \overline{c_1})$ are prespecified constants. However, we show that this formulation can be expressed as a weighted linear combination of the GPNs defined in \eqref{eq2}; see Appendix~\ref{app:interval_gpn} for details. Consequently, our discussion of GPN subsumes this estimand.

{\bf Connection to Previous Definitions}. GPN provides a unified way for defining PN with binary and ordinal outcomes. Specifically, 
(a) binary outcomes: for $Y \in \{0, 1\}$ with thresholds $c_0=c_1=1$, GPN reduces to the PN in \eqref{eq1}.   
(b) Ordinal Outcomes: For $Y \in {1, \ldots, K}$,~\cite{Zhang2024} defined the PNS for a specific outcome level as $\mathbb{P}(Y(0) < y \mid Z = 1, Y = y)$. When a restriction is imposed (e.g., $y \le Y < y + \delta$ for a small $\delta$), the flexible definition of the GPN naturally encompasses this estimand.

{\bf Challenges.} Two main challenges arise in identifying GPN.

(a) \emph{Joint distribution of potential outcomes}. The definition of GPN depends on the joint distribution of potential outcomes. Since the joint values of potential outcomes are never observed simultaneously for each individual, the estimand is inherently unidentifiable from the observed data, even under randomization. 

(b) \emph{Standard conditions for binary and ordinal outcomes are insufficient}. For identification, beyond the standard ignorability assumption~(Assumption~\ref{asp1}), common  conditions for discrete outcomes are inadequate for GPN.  Specifically, assumptions such as monotonicity (for binary outcomes, Assumption~\ref{asp2}) and the monotonic incremental treatment effect (for ordinal outcomes)~\cite{Zhang2024} do not yield point identification for continuous outcomes. 

Consequently, rather than imposing additional restrictive and untestable assumptions to force point identification, we adopt a partial identification approach that focuses on deriving sharp bounds on the GPN. 
The sharp bounds represent the tightest bounds attainable for an estimand under the observed data distribution and maintained assumptions~\cite{Fan-Park-2010, Rosenbaum2020, Ding2023}.

\section{Sharp Bounds under Common Conditions}
\label{sec:common}
In this section, we explore the sharp bounds on the conditional GPN,  defined as
\begin{equation} 
	\label{eq:gpn_target}
	\begin{split}
		\textup{GPN}&(x) ={}  \mathbb{P}(Y(0) < c_0 \mid X=x, Z=1, Y \ge c_1) \\
		={}& \frac{\mathbb{P}(Y(0) < c_0, Y(1) \ge c_1 \mid X=x, Z=1)}{\mathbb{P}(Y \ge c_1 \mid X=x, Z=1)}. 
	\end{split}
\end{equation}
We derive these bounds under the standard ignorability assumption and under the combined ``ignorability + monotonicity'' assumptions, respectively.

Analyzing the sharp bounds on $\text{GPN}(x)$ is more general than studying GPN directly, since we can always obtain the sharp bounds on GPN by taking the expectation of the sharp bounds on $\text{GPN}(X)$  with respect to $X$. Throughout, we maintain the standard ignorability assumption.

\begin{assumption}[Ignorability]\label{asp1} \, \\
	(a) Unconfoundedness: $(Y(1), Y(0)) \indep Z \mid X$;\\  
	(b) Overlap: $0 <  \mathbb{P}(Z=1 \mid X=x) < 1$ for all $x \in \mathcal{X}$. 
\end{assumption}

The unconfoundedness states that, conditional on covariates, treatment assignment is independent of the potential outcomes, while the overlap requires that every individual has a positive probability of receiving each treatment. Under Assumption~\ref{asp1}, the conditional marginal cumulative distribution functions (CDFs)
\( F_{Y(1)|X}(y|x) = \mathbb{P}(Y(1) \le y \mid X = x) \) and
\( F_{Y(0)|X}(y|x) = \mathbb{P}(Y(0) \le y \mid X = x) \)
are identifiable from the observed data. However, Assumption~\ref{asp1} alone is insufficient to identify \( \text{GPN}(x) \), since it depends on the joint distribution of \( (Y(1), Y(0)) \), which is not identifiable.

\subsection{Sharp Bounds under Ignorability}
\label{sec:general_bounds}
We first derive sharp bounds on $\text{GPN}(x)$ that rely solely on the ignorability assumption. These bounds build upon the Fréchet--Hoeffding  inequalities~\cite{Frechet1960}. 

\begin{lemma}[Fréchet--Hoeffding Bounds]
	\label{thm:fh_bounds}
	Under Assumption \ref{asp1}, $\textup{GPN}(x)$ lies within the interval $[L_{\textup{FH}}(x), U_{\textup{FH}}(x)]$, where 
	\begin{align*}
		& L_{\textup{FH}}(x)={} \max
		\left(0,  1 -\frac{ \mathbb{P}(Y \ge c_0 \mid X=x, Z=0)}{\mathbb{P}(Y \ge c_1 \mid X=x, Z=1)}
		\right), \\
		&U_{\textup{FH}}(x) = \min\left(1, \frac{\mathbb{P}(Y < c_0 \mid X=x, Z=0)}{\mathbb{P}(Y \ge c_1 \mid X=x, Z=1)}\right).
	\end{align*}
	These bounds are sharp. 
\end{lemma}

The Fréchet--Hoeffding bounds in Lemma~\ref{thm:fh_bounds} require fewer 
assumptions, but they are often wide in practice. This is because Assumption~\ref{asp1} imposes no restriction on the relationship between \( Y(1) \) and \( Y(0) \). Consequently, the Fréchet--Hoeffding bounds
represent the widest possible limits, capturing the worst- and best-case scenarios for an arbitrary dependence between \( Y(1) \) and \( Y(0) \).

\subsection{Sharp Bounds under Ignorability and Monotonicity}
For binary outcomes, monotonicity is a standard assumption for identifying PN~\cite{Pearl1999}, but it is insufficient to identify GPN for continuous outcomes. We examine the bounds on $\text{GPN}(x)$ under this additional assumption.

\begin{assumption}[Monotonicity] \label{asp2} 
	$\mathbb{P}(Y(1) \geq Y(0)) = 1$.
\end{assumption}

Monotonicity assumes that the treatment has a non-negative effect for all individuals almost surely. While commonly imposed, this assumption is restrictive and hard to satisfy in settings with substantial individual heterogeneity.

\begin{proposition}[Sharp Bounds under Monotonicity]
	\label{lemma:mono_bounds}
	Under Assumptions \ref{asp1} and \ref{asp2}, 
	$\textup{GPN}(x)$ lies within the interval $[L_{\textup{mono}}(x), U_{\textup{mono}}(x)]$, where 
	\begin{align*}
		&L_{\textup{mono}}(x) = L_{\textup{FH}}(x), \\
		&U_{\textup{mono}}(x)= \min \left( U_{\textup{FH}}(x),    1-\frac{\mathbb{P}(Y \ge c_1 \mid X=x,Z=0)}{\mathbb{P}(Y \ge c_1 \mid X=x, Z=1)}\right ). 
	\end{align*}
	These bounds are sharp.
\end{proposition}

Proposition \ref{lemma:mono_bounds} establishes the sharp bounds on $\text{GPN}(x)$ under the ignorability and  monotonicity assumptions. Comparing Proposition~\ref{lemma:mono_bounds} with Lemma~\ref{thm:fh_bounds} highlights the role that monotonicity plays in tightening the bounds on $\text{GPN}(x)$.

Specifically, the lower bound remains unchanged, indicating that the monotonicity does not contribute additional information for the lower bound on $\text{GPN}(x)$. The reason is that monotonicity permits the case $Y(1)=Y(0)$, under which $\text{GPN}(x)$ is equal to zero. 
In contrast, the upper bound is tightened, i.e.,  $U_{\text{mono}}(x) \le U_{\text{FH}}(x)$. 
Intuitively, this improvement arises from excluding individuals for whom $Y(0) > Y(1)$.  
Moreover, the $\text{GPN}(x)$ is point-identifiable when $c_1 = c_0$ under the ignorability and monotonicity assumptions, as shown in Corollary \ref{lemma:mono_bounds_same}.

\begin{corollary}[Identification under Monotonicity]
	\label{lemma:mono_bounds_same}
	Under Assumptions \ref{asp1} and \ref{asp2}, if $c_0 = c_1 = c$, $\textup{GPN}(x)$ is point-identifiable, and the identifiability formula is given by 
	\begin{equation*}
		\textup{GPN}(x) = 1 - \frac{\mathbb{P}(Y(0) \ge c \mid X=x)}{\mathbb{P}(Y(1) \ge c \mid X=x)}.
	\end{equation*}
\end{corollary}

Corollary~\ref{lemma:mono_bounds_same} aligns with our intuition. When $c_1= c_0 = c$, the continuous outcomes is collapsed into a binary variable, and the problem reduces to the binary outcomes setting studied in~\cite{Pearl1999}. 
However, when $c_1 > c_0$, the equivalence no longer holds. Although monotonicity restricts the direction of the treatment effect, it leaves the \emph{strength} of dependence between $Y(1)$ and $Y(0)$ largely unspecified. 
This suggests that obtaining sharper bounds requires capturing the continuous dependence structure.
We achieve this using copula models in Section~\ref{sec:copula}.

\section{Improved Bounds via Copula-Based Sensitivity Analysis}
\label{sec:copula}

As discussed in Section~\ref{sec:general_bounds}, the Fréchet–-Hoeffding bounds are often too wide to be informative, whereas the monotonicity is restrictive and provides only limited improvement in tightening the bounds. In this section, we propose an alternative approach that improves the bounds by explicitly modeling the dependence between the potential outcomes.

\subsection{Copula Representation of GPN}
\label{subsec:copula_rep}

For ease of presentation, we define 
\[
\begin{aligned}
	u_1(x) = \mathbb{P}(Y(1) \le c_1 \mid x),~u_0(x) = \mathbb{P}(Y(0) \le c_0 \mid x),
\end{aligned}
\]
where are identifiable under Assumption \ref{asp1}.  
According to Sklar's Theorem~\cite{copulas2006,sklar1959fonctions}, for any continuous joint distribution, there exists a unique copula $C_\rho: [0,1]^2 \to [0,1]$ such that
\[
\mathbb{P}(Y(1) \le c_1, Y(0) \le c_0 \mid x) = C_\rho(u_1(x), u_0(x)),
\]
where $\rho$ is the association parameter that governs the dependence between the potential outcomes.  
This copula representation enables an explicit expression of $\textup{GPN}(x)$ in terms of the marginal distributions and the parameter~$\rho$.

\begin{proposition}[Copula Representation]
	\label{prop:gpn_copula}
	For a given copula family $C_{\rho}(\cdot, \cdot)$,  $\textup{GPN}(x)$ admits the following formulation: 
	\begin{align}
		\label{eq:gpn_copula_gen}
		\textup{GPN}(x; \rho) 
		= \frac{u_0(x) - C_\rho(u_1(x), u_0(x))}{1 - u_1(x)}.
	\end{align}
\end{proposition}

Equation~\eqref{eq:gpn_copula_gen} decomposes $\textup{GPN}(x)$ into two components: the marginal probabilities $u_1(x)$ and $u_0(x)$, which are point-identifiable under Assumption~\ref{asp1}, and the copula term $C_\rho$, which characterizes the unobserved dependence between the potential outcomes. Although the association parameter $\rho$ is not identifiable due to the fundamental problem of causal inference~\cite{Holland1986}, this formulation provides a natural basis for deriving bounds on $\textup{GPN}(x)$ by leveraging information about the dependence structure.

\subsection{Improved Bounds via Dependence Restriction}
\label{subsec:bounds_restriction}

Although the association parameter $\rho$ is not identifiable under Assumption~\ref{asp1}, in real-world applications, we may provide a plausible range for it, informed by domain and expert knowledge~\cite{Bodik2025}.

\begin{assumption}[Dependence Restriction]
	\label{ass:rho_bound}
	$\rho \in [\rho_{\textup{min}}, \rho_{\textup{max}}]$, where the lower and upper bounds are given.
\end{assumption}

Assumption~\ref{ass:rho_bound} restricts the strength of dependence between the potential outcomes. Under this assumption, bounds on $\rho$ translate directly into bounds on the copula value and, consequently, into bounds on $\textup{GPN}(x)$. Specifically, for an arbitrary pair $(u, v)$, let 
\begin{align*}
	C_{\textup{min}}(u, v) ={}& \min_{\rho\in [\rho_{\textup{min}}, \rho_{\textup{max}}]} C_\rho(u, v),  \\
	C_{\textup{max}}(u, v) ={}& \max_{\rho\in [\rho_{\textup{min}}, \rho_{\textup{max}}]} C_\rho(u, v),
\end{align*}
then we have the following conclusion. 

\begin{theorem}[Sharp Bounds under Copula Model]
	\label{thm:improved_bounds}
	Under Assumptions \ref{asp1} and \ref{ass:rho_bound}, the sharp bounds on $\textup{GPN}(x)$ are 
	\[
	\bigg[ \frac{u_0(x) - C_{\textup{max}}(u_1(x), u_0(x))}{1 - u_1(x)}, \frac{u_0(x) - C_{\textup{min}}(u_1(x), u_0(x))}{1 - u_1(x)} \bigg].
	\]
\end{theorem}

Theorem~\ref{thm:improved_bounds} provides a flexible framework for obtaining bounds on $\textup{GPN}(x)$ by allowing the incorporation of domain priors regarding $\rho_{\textup{min}}$ and $\rho_{\textup{max}}$. Below, we outline several practical implications.

First, Theorem~\ref{thm:improved_bounds} includes the Fréchet--Hoeffding bounds as a special case corresponding to the absence of any dependence restriction. For example, when we set $\rho \in [-1, 1]$ for the Gaussian copula, the resulting bounds reduce to the Fréchet--Hoeffding bounds presented in Lemma~\ref{thm:fh_bounds} (see Appendix~\ref{app:consistency} for details). This also  indicates how the worst-case bounds (i.e., the Fréchet--Hoeffding bounds) can be improved: any domain knowledge that restricts $\rho$ to a subset of its full range serves to tighten the identification interval.

Second, even if it may be difficult to specify exact values for $\rho_{\textup{min}}$ and $\rho_{\textup{max}}$ due to limited expert knowledge, determining the sign of the dependence is often feasible. A simple restriction is to assume positive dependence (i.e., $\rho_{\textup{min}} = 0$). This assumption is mild and reasonable in many contexts. For instance, in medical studies, unmeasured factors such as a patient's overall health or genetic resilience are likely to influence both $Y(0)$ and $Y(1)$ in the same direction, resulting in a positive correlation between the potential outcomes~\cite{efron1991compliance}. Furthermore, Wu et al.~\cite{Wu-etal-2024-Harm} and Bodik et al.~\cite{Bodik2025} provide additional examples illustrating the plausibility of assuming positive dependence.

Third, Theorem~\ref{thm:improved_bounds} naturally establishes a sensitivity analysis framework for $\textup{GPN}(x)$ by treating $\rho$ as the sensitivity parameter, enabling the assessment of the robustness of our conclusions. 
Such insights are not captured in a purely worst-case analysis, as exemplified by the Fréchet-–Hoeffding bounds and the sharp bounds derived under the monotonicity assumption. 
In addition, it reveals how expert knowledge can further refine these bounds; see our application in Section~\ref{subsec:application} for details.

\section{Empirical Implementation of the Bounds}
\label{sec:estimation}
The preceding analysis shows that the sharp bounds in Theorem~\ref{thm:improved_bounds} for $\textup{GPN}(x)$ are fully determined by the conditional distribution
functions
$u_z(x) = \mathbb{P}(Y(z) \le c_z \mid X=x)$ for $z \in \{0,1\},$ 
together with the specified dependence structure in the copula framework. 
Likewise, the bounds in Lemma~\ref{thm:fh_bounds} and Proposition~\ref{lemma:mono_bounds} can be expressed in terms of these similar  quantities. Consequently, the statistical task reduces to estimating the conditional distribution
functions.

We take $u_1(x)$ as an example to illustrate the estimation procedure; the estimation of $u_0(x)$ follows analogously. Under Assumption~\ref{asp1}, $u_1(x)$ is identified as 
$$
u_1(x) = \mathbb{E}[\mathbb{I}(Y \le c_1) \mid X=x, Z=1].
$$
Based on it, a direct approach to estimate $u_1(x)$ is to regress the indicator $\mathbb{I}(Y \le c_1)$ on $X$ using the treated units. 
To enhance robustness to model misspecification, we instead employ the doubly robust (DR) estimator for the conditional distribution function  proposed by~\cite{givens2024conditional}.
For clarity, we outline the estimation procedure for the DR estimator, which consists of two main steps:

\textbf{Step 1: Nuisance Parameter Estimation.} We first estimate the nuisance parameters for the efficient influence function in the DR estimator: the propensity score $\hat e(x)=\hat{\mathbb{P}}(Z=1\mid X=x)$ and the initial conditional distribution function $\hat\nu_1(x)=\hat{\mathbb{P}}(Y\le c_1\mid Z=1, X=x)$.
We could estimate these quantities using parametric models (e.g., linear models) or flexible machine learning methods (e.g., MLPs).

\textbf{Step 2: Pseudo-Outcome Regression.}
Using the nuisance parameter estimates from Step 1, we construct a doubly robust pseudo-outcome $\hat{\varphi}_{1,i}$ for each unit $i$.
Derived from the efficient influence function, it takes the form
\begin{equation}
	\label{eq:dr_pseudo}
	\hat{\varphi}_{1,i} = \hat{\nu}_1(X_i) + \frac{Z_i}{\hat{e}(X_i)} \left( \mathbb{I}(Y_i \le c_1) - \hat{\nu}_1(X_i) \right).
\end{equation}
We then regress the pseudo-outcomes $\{\hat{\varphi}_{1,i}\}$ on the covariates $\{X_i\}$ to obtain the estimator $\hat{u}_1(x)$.

Once the doubly robust estimates $\hat{u}_1(x)$ and $\hat{u}_0(x)$ are obtained, we could compute the sharp bounds on $\textup{GPN}(x)$ in Theorem~\ref{thm:improved_bounds} using a plug-in approach. 
For example, under the Gaussian copula framework, substituting these estimates into the analytic expressions derived in Corollary~\ref{prop:bounds_gaussian} with the specified correlation range $[\rho_{\textup{min}}, \rho_{\textup{max}}]$ yields 
\begin{align*} 
	\hat{L}_{\textup{copula}}(x) &= \frac{\hat{u}_0(x) - \Phi_{\rho_{\textup{max}}}(\Phi^{-1}(\hat{u}_1(x)), \Phi^{-1}(\hat{u}_0(x)))}{1 - \hat{u}_1(x)}, \\
	\hat{U}_{\textup{copula}}(x) &= \frac{\hat{u}_0(x) - \Phi_{\rho_{\textup{min}}}(\Phi^{-1}(\hat{u}_1(x)), \Phi^{-1}(\hat{u}_0(x)))}{1 - \hat{u}_1(x)}.
\end{align*}

\subsection{Illustration with Gaussian Copula}
\label{subsec:gaussian_illus}
We illustrate the practical application of Theorem~\ref{thm:improved_bounds} using the widely employed Gaussian copula, which is explicitly given by
$C_\rho(u, v) = \Phi_{\rho}(\Phi^{-1}(u), \Phi^{-1}(v)),$ 
where $\Phi_{\rho}$ denotes the bivariate standard normal cumulative distribution function with correlation $\rho$, and $\Phi^{-1}$ is its inverse. Since $C_\rho(u, v)$ is monotonic in $\rho$, the bounds on $\textup{GPN}(x)$ can be obtained directly by applying Theorem~\ref{thm:improved_bounds}.

\begin{corollary}[Sharp Bounds with Gaussian Copula]
	\label{prop:bounds_gaussian}
	Under Assumptions~\ref{asp1} and \ref{ass:rho_bound}, with a Gaussian copula, $\textup{GPN}(x)$ lies within the interval $[L_{\textup{copula}}(x), U_{\textup{copula}}(x)]$, where 
	\begin{align*}
		L_{\textup{copula}}(x) &= \frac{u_0(x) - \Phi_{\rho_{\textup{max}}}(\Phi^{-1}(u_1(x)), \Phi^{-1}(u_0(x)))}{1 - u_1(x)},\\
		U_{\textup{copula}}(x) &= \frac{u_0(x) - \Phi_{\rho_{\textup{min}}}(\Phi^{-1}(u_1(x)), \Phi^{-1}(u_0(x)))}{1 - u_1(x)},
	\end{align*}
	where
	$$
	\begin{aligned}
		&\Phi_{\rho}(\Phi^{-1}(u_1(x)), \Phi^{-1}(u_0(x)))\\
		=& \int_{-\infty}^{a} \int_{-\infty}^{b} \frac{1}{2 \pi \sqrt{1-\rho^2}} \exp \left(-\frac{s^2+t^2-2 \rho s t}{2\left(1-\rho^2\right)}\right) \intd s \intd t, 
	\end{aligned}
	$$
	with $a = \Phi^{-1}(u_0(x))$ and $b = \Phi^{-1}(u_1(x))$. 
\end{corollary}

Based on Corollary~\ref{prop:bounds_gaussian}, setting $\rho \in [0, 1]$ allows researchers to quantify attribution under the assumption of non-negative dependence, while narrower ranges (e.g., $\rho \in [0, 0.5]$) can be used to assess robustness to weak-to-moderate positive associations. Additionally, several other families of copulas can be employed to flexibly model the joint distribution of potential outcomes; interested readers may refer to~\cite{copulas2006, joe2014dependence} for further details.

\section{Experiments}
\label{sec:exp}
In this section, we evaluate the finite-sample performance of the proposed methods and then illustrate their practical utility through a real-world application.

\subsection{Simulation Study}
\label{subsec:sim}

We consider three data-generating processes (DGPs). In all scenarios, we generate a 3-dimensional covariate vector $X = (X_0, X_1, X_2)^\top$ from a multivariate normal distribution $\mathcal{N}(\mu_X, \Sigma_X)$, with mean $\mu_X = (0.5, -0.5, 0)^\top$ and covariance matrix
$$
\Sigma_X = \begin{bmatrix} 1 & 0.3 & 0.1 \\ 0.3 & 1 & 0.2 \\ 0.1 & 0.2 & 1 \end{bmatrix}.
$$
The binary treatment $Z$ is generated from a logistic model: $\mathbb{P}(Z=1 \mid X) = \sigma(-1.5 + 0.5 e^{X_0} - 0.8 X_1 + 0.4 X_2)$, where $\sigma(\cdot)$ is the sigmoid function. The potential outcomes are generated as $Y(z) = \mu_z(X) + \varepsilon_z$ for $z \in \{0, 1\}$.

We consider three cases for the outcome mechanisms.

\emph{(a) Monotonic Case:} 
The treatment effect is non-negative for all units. We set $\mu_0(X) = 10 + 2X_0 + X_1 - 0.5X_2$ and $\mu_1(X) = \mu_0(X) + X_0^2$. 
The errors  $\varepsilon_1 = \varepsilon_0 \sim \mathcal{N}(0, 1)$.

\emph{(b) Linear Case ($\rho = 0.5$):} 
The means are linear in $X$: $\mu_0(X) = 10 + 2X_0 + X_1 - 0.5X_2$ and $\mu_1(X) = 12 + 1.5X_0 + 1.2X_1 + 1.5X_2$. The error terms $(\varepsilon_0, \varepsilon_1)^\top$ are drawn from $\mathcal{N}(0, \Sigma_\varepsilon)$ with covariance $\Sigma_\varepsilon = \begin{bmatrix} 1 & 0.5 \\ 0.5 & 1 \end{bmatrix}$.

\emph{(c) Nonlinear Case ($\rho = 0.5$):} The outcomes exhibit complex nonlinear dependencies, with $\mu_0(X) = 10 + 0.5X_0^2 + X_1^2 + X_0X_2 - 2X_1X_2$ and $\mu_1(X) = 13 + 0.5X_0^2 + 2.5X_1^2 + 1.5X_0X_2 + 4X_1X_2$. As in the linear case, the error terms $(\varepsilon_0, \varepsilon_1)^\top$ are drawn from $\mathcal{N}(0, \Sigma_\varepsilon)$.

We set the thresholds as $c_0 = 10.0$ and $c_1 = 12.0$, and the training sample size to $n=4{,}096$.

\textbf{Implementation.} 
The marginal distributions $u_z(x) = \mathbb{P}(Y(z) \le c_z \mid X=x)$ and the monotonicity correction term $u_{0, c_1}(x) = \mathbb{P}(Y(0) \le c_1 \mid X=x)$ are estimated using the DR learner. We implement the nuisance and target estimators using Multi-Layer Perceptrons (MLPs) in PyTorch, ensuring robustness against model misspecification.

\textbf{Competing Methods}. 
We compare four bounding strategies: 
(1) \emph{FH}: The Fréchet--Hoeffding bounds under only ignorability, see Lemma \ref{thm:fh_bounds}; 
(2) \emph{Monotonicity}: the bounds under ignorability and additional monotonicity, see Proposition \ref{lemma:mono_bounds}; 
(3) \emph{Conservative}: the bounds under ignorability and a conservative restriction on $\rho$ (i.e., $\rho \in [0,1]$) under the Gaussian copula, see Theorem~\ref{thm:improved_bounds} and Corollary~\ref{prop:bounds_gaussian}; (4) \emph{Expert}: similar to the \emph{Conservative} method, we impose a more stringent restriction on $\rho$: 
$\rho \in [0.2, 0.7]$ for Cases (b)–(c) and $\rho \in [0.5, 1.0]$ for Case (a). For the Expert method, the range of $\rho$ is specified solely for ease of comparison, with the aim of examining how varying levels of prior information influence the bounds.

\textbf{Evaluation Metrics.} 
We evaluate numerical performance by computing the mean squared error (MSE) between each estimated bound (the lower bound $\hat{L}(x)$ and the upper bound $\hat{U}(x)$) and the true $\textup{GPN}(x)$: 
\begin{align*}
\text{MSE-LB} ={}& n^{-1} \sum_{i=1}^n \big( \hat{L}(X_i) - \text{GPN}(X_i)\big)^2, \\
\text{MSE-UB} ={}& n^{-1} \sum_{i=1}^n \big ( \hat{U}(X_i) - \text{GPN}(X_i) \big)^2,
\end{align*}
along with the average interval width:
$$
\text{Width} = n^{-1} \sum_{i=1}^n \left( \hat{U}(X_i) - \hat{L}(X_i) \right).$$ 
Smaller values of these metrics indicate that the estimated bounds on $\text{GPN}(x)$ are tighter. 
Consistent with the DGPs, the true $\text{GPN}(x)$ is derived analytically. 
Specifically, for case (a), 
$
\mathbb{P}(Y(1) \le c_1, Y(0) \le c_0 \mid X=x) = \Phi\left( \min(c_0, c_1 - \tau(x)) - \mu_0(x) \right)
$, where $\tau(x) = \mu_1(x) - \mu_0(x)$. 
For cases (b) and (c), 
$\mathbb{P}(Y(1) \le c_1, Y(0) \le c_0 \mid X=x) = \Phi_\rho(c_1 - \mu_1(x), c_0 - \mu_0(x))$.

\begin{table*}[t]
\caption{Performance of $\textup{GPN}(x)$ bounds across three scenarios ($n=4{,}096$).}
\label{tab:mse_results}
\begin{center}
	\begin{small}
		\begin{sc} 
			\resizebox{\textwidth}{!}{%
				\begin{tabular}{lccccccccc}
					\toprule
					\multirow{2}{*}{Method} & \multicolumn{3}{c}{(a) Monotonic Case} & \multicolumn{3}{c}{(b) Linear DGP} & \multicolumn{3}{c}{(c) Nonlinear DGP} \\
					\cmidrule(lr){2-4} \cmidrule(lr){5-7} \cmidrule(lr){8-10}
					& MSE-LB & MSE-UB & Width & MSE-LB & MSE-UB & Width & MSE-LB & MSE-UB & Width \\ 
					\midrule
					FH       & 0.015 & 0.516 & 0.547 & 0.094 & 0.218 & 0.435 & 0.057 & 0.102 & 0.079 \\
					Monotonicity            & 0.015 & 0.463 & 0.506 & 0.094 & 0.139 & 0.331 & 0.057 & 0.077 & 0.052 \\
					Conservative ($\rho \in [0,1]$) & 0.015 & 0.226 & 0.299 & 0.094 & 0.068 & 0.132 & 0.057 & 0.063 & 0.032 \\
					Expert                  & \textbf{0.015} & \textbf{0.087} & \textbf{0.159} & \textbf{0.065} & \textbf{0.060} & \textbf{0.066} & \textbf{0.056} & \textbf{0.060} & \textbf{0.016} \\ 
					\bottomrule
				\end{tabular}%
			}
		\end{sc}
	\end{small}
\end{center}
\end{table*}

\begin{figure*}[htbp]
\centering
\includegraphics[width=\textwidth]{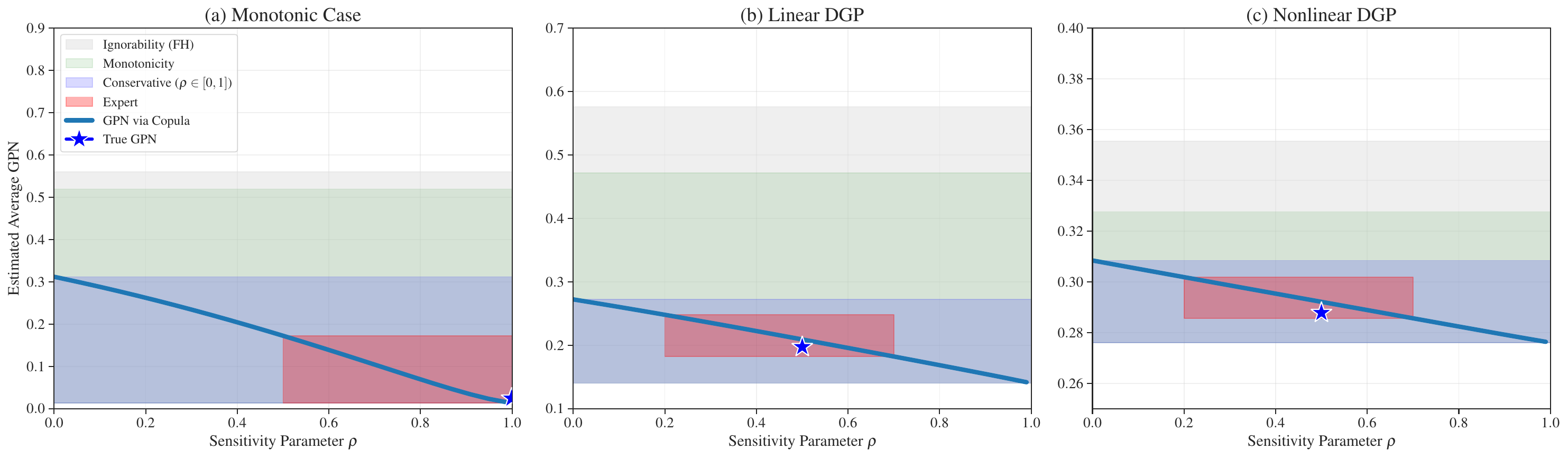}
\caption{
	Sensitivity analysis of the average GPN. The truth (star) falls within the Expert region and on the estimated Copula curve.
}
\label{fig:sensitivity}
\end{figure*}

\textbf{Numerical Results.} 
The numerical results are presented  in Table~\ref{tab:mse_results}.  First, the estimated bounds exhibit a clear performance ordering: the Expert method outperforms the Conservation method, and both outperform the Monotonicity and FH methods.
Second, the Expert method significantly outperforms the other methods, indicating that incorporating domain knowledge about $\rho$ is highly effective in narrowing the bounds on $\text{GPN}(x)$. 
Third, it is noteworthy that in case (a), where the monotonicity assumption is satisfied, the Monotonicity method still yields a wider interval (Width: 0.506) and a higher MSE-UB (0.463) than the copula-based Expert method (Width: 0.159, MSE-UB: 0.087). This suggests that in continuous outcome settings, the standard monotonicity assumption is often less informative than explicitly modeling the dependence structure.

Figure~\ref{fig:sensitivity} illustrates the estimated bounds on the average GPN, where the star indicates the true value.  
The estimated average GPN is given by $n^{-1}\sum_{i=1}^n \widehat{\textup{GPN}}(X_i)$, where $\widehat{\textup{GPN}}(X_i)$ denotes the generic estimated bounds obtained using different methods.
In Figure~\ref{fig:sensitivity}, the gray shaded region corresponds to the average FH bounds. The green shaded region shows the Monotonicity bounds. The blue and red shaded regions represent the Conservative and Expert bounds, respectively. The Conservative and Expert bounds are substantially tighter than those under Ignorability and Monotonicity, while still containing the true value.

Figure~\ref{fig:sensitivity} also presents a sensitivity analysis with respect to $\rho$, where the blue line represents the estimated average GPN under a Gaussian copula across different values of $\rho$. This explicitly illustrates how changes in $\rho$ affect the estimated GPN, highlighting the importance of choosing an appropriate value of $\rho$ to obtain tighter bounds on GPN.

\textbf{Sensitivity Analysis.} We further conduct a sensitivity analysis to assess the impact of copula family misspecification. We perform additional simulations in which the true joint distribution of $(Y(1), Y(0))$ follows a Gaussian, Clayton, or Gumbel copula, while estimation is still carried out using the Gaussian copula as a working model. We consider three levels of dependence, indexed by Kendall's $\tau \in {0.20, 0.33, 0.50}$, corresponding approximately to $\rho \in {0.31, 0.50, 0.71}$ under the Gaussian working model. We retain the same covariate distribution and treatment-assignment mechanism as in Cases (a)--(c), and generate the potential outcomes from the linear mean model: $Y(z)=\mu_z(X)+\varepsilon_z$ for $z \in {0, 1}$, where $\mu_0(X)=10+2X_0+X_1-0.5X_2$ and $\mu_1(X)=12+1.5X_0+1.2X_1+1.5X_2$, with $(\varepsilon_0,\varepsilon_1)$ generated from the specified copula family. The numerical results are reported in Table~\ref{tab:sensitivity}. Across these scenarios, the lower and upper bounds, as well as the interval widths, vary only modestly across copula families, indicating a degree of robustness.

\begin{table*}[htbp]
\caption{Sensitivity to Copula Family Misspecification.}
\label{tab:sensitivity}
\centering
\begin{sc} 
	\resizebox{\linewidth}{!}{%
		\begin{tabular}{ccccccccccccc}
			\toprule
			Family & {$\tau$} & {$\rho$} & {True Mean GPN} & {FH\_LB} & {FH\_UB} & {FH\_Width} & 
			\makecell{Conservative\\($\rho\in[0,1]$)\_LB} & 
			\makecell{Conservative\\($\rho\in[0,1]$)\_UB} & 
			\makecell{Conservative\\($\rho\in[0,1]$)\_Width} & {Expert\_LB} & {Expert\_UB} & {Expert\_Width} \\
			\midrule
			Gaussian & 0.20 & 0.31 & 0.206 & 0.152 & 0.617 & 0.464 & 0.185 & 0.435 & 0.250 & 0.318 & 0.395 & 0.077 \\
			Gaussian & 0.33 & 0.50 & 0.197 & 0.150 & 0.620 & 0.470 & 0.185 & 0.434 & 0.249 & 0.272 & 0.344 & 0.072 \\
			Gaussian & 0.50 & 0.71 & 0.186 & 0.134 & 0.606 & 0.472 & 0.168 & 0.416 & 0.249 & 0.207 & 0.280 & 0.073 \\
			Clayton  & 0.20 & 0.31 & 0.206 & 0.145 & 0.607 & 0.462 & 0.176 & 0.417 & 0.241 & 0.302 & 0.377 & 0.076 \\
			Clayton  & 0.33 & 0.50 & 0.197 & 0.154 & 0.621 & 0.467 & 0.187 & 0.437 & 0.250 & 0.273 & 0.346 & 0.073 \\
			Clayton  & 0.50 & 0.71 & 0.186 & 0.131 & 0.601 & 0.470 & 0.160 & 0.412 & 0.253 & 0.198 & 0.270 & 0.072 \\
			Gumbel   & 0.20 & 0.31 & 0.206 & 0.163 & 0.624 & 0.460 & 0.195 & 0.442 & 0.246 & 0.328 & 0.403 & 0.075 \\
			Gumbel   & 0.33 & 0.50 & 0.197 & 0.143 & 0.598 & 0.455 & 0.178 & 0.417 & 0.240 & 0.263 & 0.332 & 0.070 \\
			Gumbel   & 0.50 & 0.71 & 0.186 & 0.156 & 0.624 & 0.468 & 0.193 & 0.439 & 0.246 & 0.230 & 0.302 & 0.072 \\
			\bottomrule
		\end{tabular}%
	}
\end{sc}
\end{table*}

\subsection{Real Data Application}
\label{subsec:application}

\textbf{Data.}
We apply our proposed method to assess whether maternal healthy behaviors constitute a necessary risk factor for reductions in infant birth weight, from the perspective of attribution analysis. The dataset\footnote{It is publicly available at \url{http://www.stata-press.com/data/r13/cattaneo2.dta}.} contains $n = 4{,}642$ observations from the Pennsylvania Department of Health. The outcome $Y$ is the infant's birth weight (in grams). We define the treatment $Z$ as the adoption of healthy behaviors. We consider two risk factors as the treatment:
(1) \texttt{Smoking}: $Z = 1$ denotes non-smoking mothers, and $Z = 0$ denotes smoking mothers.
(2) \texttt{Alcohol}: $Z = 1$ denotes non-drinking mothers, and $Z = 0$ denotes mothers who consume alcohol.
We control for a rich set of covariates $X$, including maternal age, education, marital status, and race, the number of prenatal care visits, etc.

\textbf{Threshold Selection.} 
For GPN, we need to set two thresholds.
(a) Adverse Threshold ($c_0 = 2{,}500$ g): We set $c_0$ to 2,500 grams, corresponding to the internationally accepted clinical definition of low birth weight established by the World Health~\cite{world1992low}. Infants below this threshold are at a significantly higher risk of neonatal mortality and long-term health complications. 
(b) Healthy Threshold ($c_1 = 3{,}000$ g): We set $c_1$ to 3,000 grams, representing a clearly healthy birth weight.

Under this setup ($c_1 > c_0$), the $\textup{GPN}$ could be used to address a specific attribution question: \emph{``For a mother who engaged in healthy behavior (did not smoke or drink) and delivered a healthy baby ($Y \ge 3{,}000$g), what is the probability that the baby would have suffered from low birth weight ($Y(0) < 2{,}500$g) had she engaged in the unhealthy behavior?''}
This query investigates whether the healthy behavior was \emph{necessary} to prevent a drastic decline in health status , rather than causing a minor weight reduction.

In this real-world application, the positive dependence between potential outcomes is biologically plausible:  unobserved factors such as maternal genetics, nutritional absorption capacity, and overall physiological resilience are likely to positively influence birth weight, regardless of maternal smoking or drinking status~\cite{almond2005costs}. 
Consistent with the simulation study, we follow the same procedure to estimate the bounds on the average GPN.

\begin{figure}[htbp]
\centering
\includegraphics[width=\linewidth]{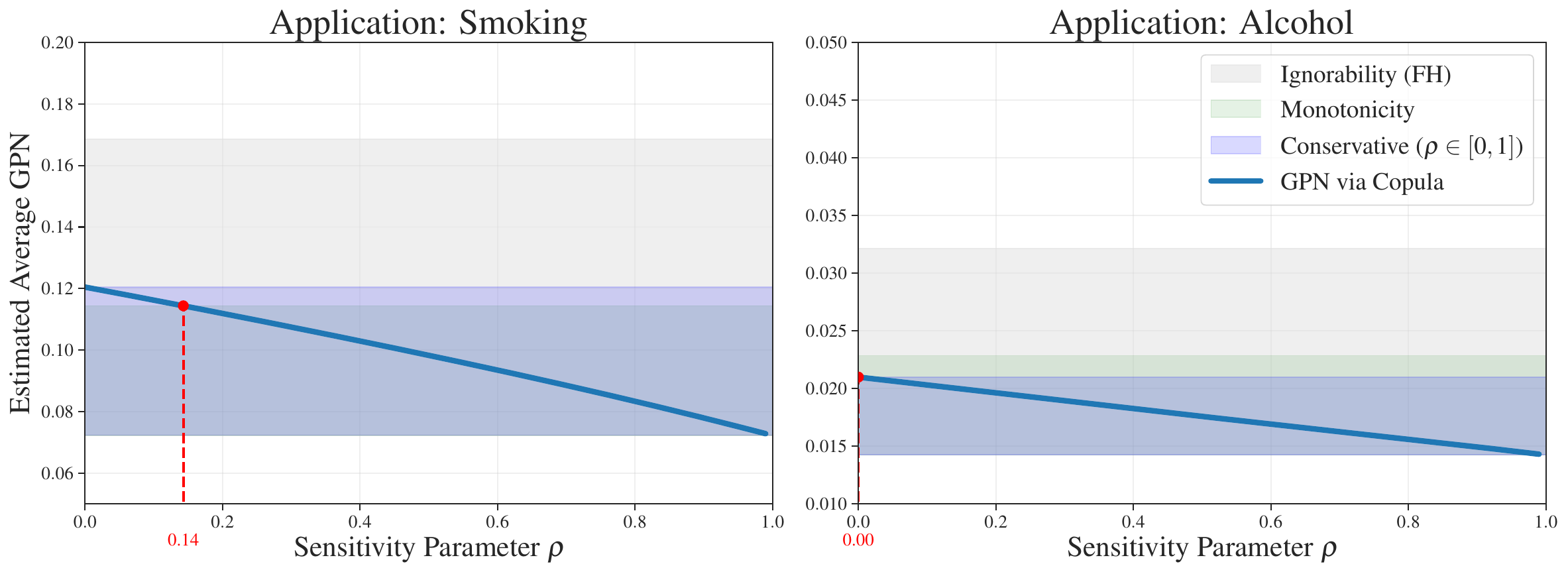}
\caption{
	Sensitivity analysis of the estimated average GPN for Smoking (left) and Alcohol (right). 
}  
\label{fig:app}
\end{figure}

\begin{table}[htbp]
\caption{Estimated bounds for the average GPN in the application.}
\label{tab:app_results}
\centering
\begin{sc} 
	\resizebox{0.85\linewidth}{!}{%
		\begin{tabular}{lcccc}
			\toprule
			\multirow{2}{*}{Method} & \multicolumn{2}{c}{Smoking} & \multicolumn{2}{c}{Alcohol} \\
			\cmidrule(lr){2-3} \cmidrule(lr){4-5}
			& Lower Bound & Upper Bound & Lower Bound & Upper Bound \\
			\midrule
			FH & $0.063 \pm 0.021$ & $0.143 \pm 0.038$ & $0.015 \pm 0.005$ & $0.045 \pm 0.011$ \\
			Monotonicity & $0.063 \pm 0.021$ & $0.088 \pm 0.026$ & $0.015 \pm 0.005$ & $0.024 \pm 0.006$ \\
			Conservative & $0.063 \pm 0.021$ & $0.097 \pm 0.029$ & $0.015 \pm 0.005$ & $0.022 \pm 0.006$ \\
			\bottomrule
		\end{tabular}%
	}
\end{sc}
\end{table}

\textbf{Results.}
The estimated bounds on GPN are visualized in Figure~\ref{fig:app}, which displays the results for two  risk factors: \texttt{Smoking} and \texttt{Alcohol}.
Similar to Figure~\ref{fig:app}, the shaded regions represent the bounds derived from the FH, Monotonicity, and Conservative methods. 
The blue line represents the sensitivity analysis with respect to $\rho$ based on the Gaussian copula. 
We mark the intersection between the curve and the tightest upper bound with a red dot. Taking \texttt{Smoking} (left panel) as an example, this figure illustrates how expert knowledge can further refine the bounds. If expert knowledge indicates that the correlation coefficient $\rho > 0.14$, we can tighten the upper bound; moreover, as long as $\rho < 1$, we can also improve the lower bound.

Table~\ref{tab:app_results} presents the numerical results for the bounds, along with their standard deviations obtained via subsampling (100 subsamples). We use subsampling rather than the ordinary bootstrap because inference for partially identified bounds involving min/max operators is nonstandard, and subsampling has well-known theoretical advantages in such irregular settings; see, for example,~\cite{romano1999}. The results reveal two key insights. First, for both \texttt{Smoking} and \texttt{Alcohol}, the lower bounds of the GPN are significantly positive. This suggests that avoiding unhealthy behaviors is a necessary cause of healthy infant birth weight for a substantial fraction of the population; that is, without such avoidance, these infants would have fallen into the low birth weight category. 
Second, compared with \texttt{Alcohol}, the bounds on the GPN for \texttt{Smoking} are larger, underscoring the severity of smoking as a risk factor. 
This indicates that avoiding smoking is critical to prevent low infant birth weight. This finding aligns with and supports for 
the public health consensus that smoking poses a greater risk to infant development and therefore warrants prioritized intervention~\cite{da2008maternal, walker2009teen}.

\section{Conclusion}
\label{conclusion}
This paper addresses a gap in attributional causal inference by developing a principled framework for the probability of necessity with continuous outcomes. We introduced the general probability of necessity (GPN), which unifies and extends existing formulations for binary and ordinal outcomes while avoiding information loss induced by discretization. Recognizing that point identification is generally impossible in this setting, we adopted a partial identification perspective and derived sharp bounds for the GPN under standard ignorability and monotonicity assumptions.
To move beyond the often overly conservative Fréchet--Hoeffding bounds, we further proposed a copula-based framework that exploits restrictions on the dependence between potential outcomes. This representation reveals that the GPN depends on a single association parameter governing the joint distribution of the potential outcomes, enabling both principled incorporation of domain knowledge and transparent sensitivity analysis. Even weak and practically plausible restrictions, such as assuming positive dependence, can substantially tighten the resulting bounds.
We hope this framework will encourage further research on dependence-aware causal attribution and on practical strategies for eliciting and validating domain knowledge about counterfactual dependence.


\bibliographystyle{plainnat}
\bibliography{references20260401}

\newpage

\begin{center}
\bf \Large 
Appendix
\end{center}

\bigskip 

\setcounter{equation}{0}
\setcounter{section}{0}
\setcounter{figure}{0}
\setcounter{example}{0}
\setcounter{proposition}{0}
\setcounter{corollary}{0}
\setcounter{theorem}{0}
\setcounter{table}{0}
\setcounter{condition}{0}
\setcounter{lemma}{0}
\setcounter{remark}{0}
\setcounter{definition}{0}

\renewcommand {\thepage} {A\arabic{page}}
\renewcommand {\theproposition} {A\arabic{proposition}}
\renewcommand {\theexample} {A\arabic{example}}
\renewcommand {\thefigure} {A\arabic{figure}}
\renewcommand {\thetable} {A\arabic{table}}
\renewcommand {\theequation} {A\arabic{equation}}
\renewcommand {\thelemma} {A\arabic{lemma}}
\renewcommand {\thesection} {A\arabic{section}}
\renewcommand {\thetheorem} {A\arabic{theorem}}
\renewcommand {\thecorollary} {A\arabic{corollary}}
\renewcommand {\thecondition} {A\arabic{condition}}
\renewcommand {\theremark} {A\arabic{remark}}
\renewcommand{\thedefinition}{A\arabic{definition}}

\setcounter{page}{1}


\section{Generalization of GPN to Interval Events}
\label{app:interval_gpn}
In Section \ref{concept}, we defined the General Probability of Necessity (GPN): 
\begin{equation*}
	\textup{GPN} = \textup{GPN}(c_0, c_1) = \mathbb{P}(Y(0) < c_0 \mid Z=1, Y \ge c_1).
	\label{eq:gpn_std}
\end{equation*}

We consider a flexible definition of GPN, denoted as $\textup{GPN}_{\textup{int}}$, defined over intervals $[\underline{c_0}, \overline{c_0})$ for $Y(0)$ and $[\underline{c_1}, \overline{c_1})$ for $Y(1)$:
\begin{equation}
	\textup{GPN}_{\textup{int}} = \mathbb{P}(\underline{c_0} \le Y(0) < \overline{c_0} \mid Z=1, \underline{c_1} \le Y(1) < \overline{c_1}).
\end{equation}

By the definition of conditional probability, we have
\begin{equation}
	\textup{GPN}_{\textup{int}} = \frac{\mathbb{P}(\underline{c_0} \le Y(0) < \overline{c_0}, \underline{c_1} \le Y(1) < \overline{c_1} \mid Z=1)}{\mathbb{P}(\underline{c_1} \le Y(1) < \overline{c_1} \mid Z=1)}.
	\label{eq:int_fraction}
\end{equation}

The denominator
in \eqref{eq:int_fraction} decomposes into
\begin{equation*}
	\mathbb{P}(\underline{c_1} \le Y(1) < \overline{c_1} \mid Z=1) = \mathbb{P}(Y(1) \ge \underline{c_1} \mid Z=1) - \mathbb{P}(Y(1) \ge \overline{c_1} \mid Z=1).
\end{equation*}

The numerator involves the joint probability of two interval events. 
Expanding the joint probability in the numerator:
\begin{align*}
	&\quad \ \mathbb{P}(\underline{c_0} \le Y(0) < \overline{c_0}, \underline{c_1} \le Y(1) < \overline{c_1} \mid Z=1) \\
	&= \mathbb{P}(\{\underline{c_0} \le Y(0) < \overline{c_0}\} \cap \{\underline{c_1} \le Y(1) < \overline{c_1}\} \mid Z=1) \nonumber \\
	&= \mathbb{P}(\{Y(0) < \overline{c_0}\} \cap \{\underline{c_1} \le Y(1) < \overline{c_1}\} \mid Z=1) \nonumber \\
	&\quad - \mathbb{P}(\{Y(0) < \underline{c_0}\} \cap \{\underline{c_1} \le Y(1) < \overline{c_1}\} \mid Z=1)\\
	&=\left[ \mathbb{P}(Y(0) < \overline{c_0}, Y(1) \ge \underline{c_1} \mid Z=1) - \mathbb{P}(Y(0) < \overline{c_0}, Y(1) \ge \overline{c_1} \mid Z=1) \right] \nonumber \\
	&\quad - \left[ \mathbb{P}(Y(0) < \underline{c_0}, Y(1) \ge \underline{c_1} \mid Z=1) - \mathbb{P}(Y(0) < \underline{c_0}, Y(1) \ge \overline{c_1} \mid Z=1) \right].
\end{align*}
Using the relation $\mathbb{P}(Y(0) < c_0, Y(1) \ge c_1 \mid Z=1) = \textup{GPN}(c_0, c_1) \cdot \mathbb{P}(Y(1) \ge c_1 \mid Z=1)$, the numerator becomes
\begin{align*}
	\mathbb{P}(\underline{c_0} \le Y(0) < \overline{c_0}, \underline{c_1} \le Y(1) < \overline{c_1} \mid Z=1) &= \textup{GPN}(\overline{c_0}, \underline{c_1}) \cdot \mathbb{P}(Y(1) \ge \underline{c_1} \mid Z=1) \nonumber \\
	&\quad - \textup{GPN}(\overline{c_0}, \overline{c_1}) \cdot \mathbb{P}(Y(1) \ge \overline{c_1} \mid Z=1) \nonumber \\
	&\quad - \textup{GPN}(\underline{c_0}, \underline{c_1}) \cdot \mathbb{P}(Y(1) \ge \underline{c_1} \mid Z=1) \nonumber \\
	&\quad + \textup{GPN}(\underline{c_0}, \overline{c_1}) \cdot \mathbb{P}(Y(1) \ge \overline{c_1} \mid Z=1).
\end{align*}

Substituting these into \eqref{eq:int_fraction} yields the linear combination:
\begin{align}
	\textup{GPN}_{\textup{int}} &= w_1 \cdot \left[ \textup{GPN}(\overline{c_0}, \underline{c_1}) - \textup{GPN}(\underline{c_0}, \underline{c_1}) \right] \nonumber \\
	&\quad - w_2 \cdot \left[ \textup{GPN}(\overline{c_0}, \overline{c_1}) - \textup{GPN}(\underline{c_0}, \overline{c_1}) \right],
\end{align}
where the weights are given by
\begin{align}
	w_1 &= \frac{\mathbb{P}(Y(1) \ge \underline{c_1} \mid Z=1)}{\mathbb{P}(Y(1) \ge \underline{c_1} \mid Z=1) - \mathbb{P}(Y(1) \ge \overline{c_1} \mid Z=1)}, \\
	w_2 &= \frac{\mathbb{P}(Y(1) \ge \overline{c_1} \mid Z=1)}{\mathbb{P}(Y(1) \ge \underline{c_1} \mid Z=1) - \mathbb{P}(Y(1) \ge \overline{c_1} \mid Z=1)}.
\end{align}
It follows that $w_1 - w_2 = 1$. 
This derivation confirms that any interval-based probability of necessity can be calculated as a linear combination of the GPNs, with weights determined by the identifiable marginal distribution of $Y(1)$.

\section{Proof of Lemma~\ref{thm:fh_bounds}}
\label{appx1}
\begin{lemma}[Fréchet--Hoeffding Bounds on GPN]
	Under Assumption \ref{asp1}, $\textup{GPN}(x)$ lies within the interval $[L_{\textup{FH}}(x), U_{\textup{FH}}(x)]$, where 
	\begin{align*}
		& L_{\textup{FH}}(x)={} \max
		\left(0,  1 -\frac{ \mathbb{P}(Y \ge c_0 \mid X=x, Z=0)}{\mathbb{P}(Y \ge c_1 \mid X=x, Z=1)}
		\right), \\
		&U_{\textup{FH}}(x) = \min\left(1, \frac{\mathbb{P}(Y < c_0 \mid X=x, Z=0)}{\mathbb{P}(Y \ge c_1 \mid X=x, Z=1)}\right).
	\end{align*}
	These bounds are sharp. 
\end{lemma}

\begin{proof}[Proof of Lemma 1]
	The $\textup{GPN}(x)$ is given by
	\begin{equation*}
		\textup{GPN}(x) = \frac{\mathbb{P}(Y(0) < c_0, Y(1) \ge c_1 \mid X=x)}{\mathbb{P}(Y(1) \ge c_1 \mid X=x)}.
	\end{equation*}
	Let $A = \{Y(0) < c_0\}$ and $B = \{Y(1) \ge c_1\}$. 
	By the Fréchet--Hoeffding inequality~\cite{Frechet1960}, the joint probability is bounded by
	\begin{equation}
		\max(0, \mathbb{P}(A \mid X=x) + \mathbb{P}(B \mid X=x) - 1) \le \mathbb{P}(A \cap B \mid X=x) \le \min(\mathbb{P}(A \mid X=x), \mathbb{P}(B \mid X=x)).
	\end{equation}
	Dividing by the denominator $\mathbb{P}(B \mid X=x)$ yields the upper bound:
	\begin{align*}
		U_{\textup{FH}}(x) &= \frac{\min(\mathbb{P}(A \mid X=x), \mathbb{P}(B \mid X=x))}{\mathbb{P}(B \mid X=x)} = \min\left(1, \frac{\mathbb{P}(A \mid X=x)}{\mathbb{P}(B \mid X=x)}\right) \\
		&= \min\left(1, \frac{\mathbb{P}(Y < c_0 \mid X=x, Z=0)}{\mathbb{P}(Y \ge c_1 \mid X=x, Z=1)}\right).
	\end{align*}
	Similarly, the lower bound is derived as
	\begin{align*}
		L_{\textup{FH}}(x) &= \frac{\max(0, \mathbb{P}(A \mid X=x) + \mathbb{P}(B \mid X=x) - 1)}{\mathbb{P}(B \mid x)} = \max\left(0, 1 - \frac{1 - \mathbb{P}(A \mid X=x)}{\mathbb{P}(B \mid X=x)}\right) \\
		&= \max\left(0, 1 - \frac{\mathbb{P}(Y \ge c_0 \mid X=x, Z=0)}{\mathbb{P}(Y \ge c_1 \mid X=x, Z=1)}\right).
	\end{align*}
	Sharpness holds because the Fréchet--Hoeffding bounds are pointwise sharp for any pair of marginal distributions.
	
	This completes the proof.
\end{proof}

\section{Proof of Proposition~\ref{lemma:mono_bounds}}
\label{appx2}
\begin{proposition}
	[Sharp Bounds under Monotonicity]
	Under Assumptions \ref{asp1} and \ref{asp2}, 
	$\textup{GPN}(x)$ lies within the interval $[L_{\textup{mono}}(x), U_{\textup{mono}}(x)]$, where 
	\begin{align*}
		&L_{\textup{mono}}(x) = L_{\textup{FH}}(x), \\
		&U_{\textup{mono}}(x)= \min \left( U_{\textup{FH}}(x),    1-\frac{\mathbb{P}(Y \ge c_1 \mid X=x,Z=0)}{\mathbb{P}(Y \ge c_1 \mid X=x, Z=1)}\right ). 
	\end{align*}
	These bounds are sharp.
\end{proposition}

\begin{proof}[Proof of Proposition 1]
	We suppress the conditioning on $X=x$ for notational brevity. Let $A = \{Y(0) < c_0\}$ and $B = \{Y(1) \ge c_1\}$. The numerator of the GPN is $\mathbb{P}(A \cap B)$. Define the event $C = \{Y(0) \ge c_1\}$.

	Under Assumption \ref{asp2} ($Y(1) \ge Y(0)$), the event $C$ implies $Y(1) \ge c_1$, so $C \subseteq B$. Additionally, since $c_1 > c_0$, $C$ implies $Y(0) \ge c_0$, so $C \subseteq A^c$. Thus, $C \subseteq B \cap A^c = B \setminus A$.

	Using the decomposition $\mathbb{P}(A \cap B) = \mathbb{P}(B) - \mathbb{P}(B \setminus A)$, the inclusion $C \subseteq B \setminus A$ implies
	\begin{equation}
		\mathbb{P}(A \cap B) \le \mathbb{P}(B) - \mathbb{P}(C).
	\end{equation}
	Combining this with the Fréchet--Hoeffding upper bound $\mathbb{P}(A)$, the sharp upper bound for the numerator is $\min(\mathbb{P}(A), \mathbb{P}(B) - \mathbb{P}(C))$. Dividing by the denominator $\mathbb{P}(B)$ yields
	\begin{align*}
		U_{\textup{mono}}(x) &= \min\left( \frac{\mathbb{P}(A)}{\mathbb{P}(B)}, \frac{\mathbb{P}(B) - \mathbb{P}(C)}{\mathbb{P}(B)} \right) \\
		&= \min\left( U_{\textup{FH}}(x), 1 - \frac{\mathbb{P}(Y \ge c_1 \mid Z=0)}{\mathbb{P}(Y \ge c_1 \mid Z=1)} \right).
	\end{align*}
	For the lower bound, the event $A \cap B$ corresponds to $Y(0) < c_0 < c_1 \le Y(1)$, which is consistent with the monotonicity assumption $Y(1) \ge Y(0)$. Thus, monotonicity imposes no additional constraints on the lower bound beyond the marginal distributions, yielding $L_{\textup{mono}}(x) = L_{\textup{FH}}(x)$.
	
	This completes the proof.
\end{proof}

\section{Proof of Corollary~\ref{lemma:mono_bounds_same}}
\label{appx3}
\begin{corollary}[Identification under Monotonicity]
	Under Assumptions \ref{asp1} and \ref{asp2}, if $c_0 = c_1 = c$, $\textup{GPN}(x)$ is point-identifiable, and the identifiability formula is given by 
	\begin{equation*}
		\textup{GPN}(x) = 1 - \frac{\mathbb{P}(Y(0) \ge c \mid X=x)}{\mathbb{P}(Y(1) \ge c \mid X=x)}.
	\end{equation*}
\end{corollary}

\begin{proof}[Proof of Corollary 1]
	Define the events $E_1 = \{Y(1) \ge c\}$ and $E_0 = \{Y(0) \ge c\}$. The numerator of $\textup{GPN}(x)$ corresponds to the joint probability $\mathbb{P}(E_1 \cap E_0^c \mid X=x)$.

	Under Assumption \ref{asp2} ($Y(1) \ge Y(0)$), observing $Y(0) \ge c$ implies $Y(1) \ge c$. Thus, $E_0 \subseteq E_1$.
	Consequently, the joint probability simplifies to
	\begin{equation}
		\mathbb{P}(E_1 \cap E_0^c \mid X=x) = \mathbb{P}(E_1 \setminus E_0 \mid X=x) = \mathbb{P}(E_1 \mid X=x) - \mathbb{P}(E_0 \mid X=x).
	\end{equation}
	Dividing by the denominator $\mathbb{P}(E_1 \mid X=x)$ yields
	\begin{equation*}
		\textup{GPN}(x) = \frac{\mathbb{P}(E_1 \mid X=x) - \mathbb{P}(E_0 \mid X=x)}{\mathbb{P}(E_1 \mid X=x)} = 1 - \frac{\mathbb{P}(Y(0) \ge c \mid X=x)}{\mathbb{P}(Y(1) \ge c \mid X=x)}.
	\end{equation*}
	
	This completes the proof.
\end{proof}

\section{Consistency with Fréchet--Hoeffding Bounds}
\label{app:consistency}
We show that the bounds derived from the Gaussian copula with $\rho \in [-1, 1]$ coincide with the Fréchet--Hoeffding bounds.

Recall the copula representation of the $\textup{GPN}(x; \rho)$in Proposition~\ref{prop:gpn_copula}:
\begin{equation*}
	\textup{GPN}(x; \rho) = \frac{u_0(x) - C_\rho(u_1(x), u_0(x))}{1 - u_1(x)},
\end{equation*}
where $u_1(x) = \mathbb{P}(Y(1) \le c_1 \mid X=x)$, $u_0(x) = \mathbb{P}(Y(0) \le c_0 \mid X=x)$, and $C_\rho$ is the copula function governing the joint distribution.

The Gaussian copula $C_\rho(u, v)$ satisfies the limit properties:
\begin{align}
	\lim_{\rho \to 1} C_\rho(u, v) &= \min(u, v), \label{eq:pos_lim}\\
	\lim_{\rho \to -1} C_\rho(u, v) &= \max(u + v - 1, 0). \label{eq:neg_lim}
\end{align}

The lower bound of $\textup{GPN}(x; \rho)$ is attained when $C_\rho$ is maximized ($\rho \to 1$). Substituting \eqref{eq:pos_lim}
\begin{equation*}
	L_{\textup{copula}}(x) = \frac{u_0(x) - \min(u_1(x), u_0(x))}{1 - u_1(x)}.
\end{equation*}

If $u_1(x) \ge u_0(x)$, the numerator equals 0; if $u_1(x) < u_0(x)$, it equals $u_0(x) - u_1(x)$. Thus
\begin{equation*}
	L_{\textup{copula}}(x) = \max\left(0, \frac{u_0(x) - u_1(x)}{1 - u_1(x)}\right).
\end{equation*}

This is identical to the Fréchet--Hoeffding lower bound $L_{\textup{FH}} = \max
\left(0,  1 -\frac{ \mathbb{P}(Y \ge c_0 \mid X=x, Z=0)}{\mathbb{P}(Y \ge c_1 \mid X=x, Z=1)}
\right)$ derived in Lemma \ref{thm:fh_bounds}.

The upper bound is attained when $C_\rho$ is minimized ($\rho \to -1$). Substituting \eqref{eq:neg_lim}
\begin{equation*}
	U_{\textup{copula}}(x) = \frac{u_0(x) - \max(u_1(x) + u_0(x) - 1, 0)}{1 - u_1(x)}.
\end{equation*}

If $u_1(x) + u_0(x) - 1 < 0$, the expression simplifies to $\frac{u_0(x)}{1 - u_1(x)}$. Otherwise, the numerator becomes $u_0(x) - (u_1(x) + u_0(x) - 1) = 1 - u_1(x)$, yielding a value of 1. Thus
\begin{equation}
	U_{\textup{copula}}(x) = \min\left(1, \frac{u_0(x)}{1 - u_1(x)}\right).
\end{equation}
This is identical to the Fréchet--Hoeffding upper bound $U_{\textup{FH}} = \min\left(1, \frac{\mathbb{P}(Y < c_0 \mid X=x, Z=0)}{\mathbb{P}(Y \ge c_1 \mid X=x, Z=1)}\right)$ derived in Lemma \ref{thm:fh_bounds}.

The equivalence establishes that the proposed copula framework generalizes the classical Fréchet--Hoeffding bounds. By restricting the association parameter (e.g., $\rho \in [0, 0.5]$) based on domain knowledge, the framework strictly tightens the identification interval compared to the worst-case bounds.

\end{document}